\pdfoutput=1
\documentclass[a4paper,reprint,twocolumn,notitlepage,aip,nofootinbib]{revtex4-2}
\usepackage[left=1.5cm,right=1.5cm,top=1cm,bottom=1.5cm,includeheadfoot]{geometry}

\usepackage{tgtermes} 
\usepackage{tgheros} 
\usepackage[T1]{fontenc}

\usepackage{amssymb}
\usepackage{amsmath}
\usepackage{amsthm}
\usepackage[dvipsnames]{xcolor}
\usepackage{soul}
\usepackage{float}
\usepackage{multirow}
\usepackage{tikz}
\usetikzlibrary{positioning,arrows}

\PassOptionsToPackage{hyphens}{url} 
\usepackage[breaklinks=true]{hyperref}
\bibpunct{[}{]}{;}{n}{}{} 

\newtheorem{theorem}{Theorem}

\newtheorem{corollary}[theorem]{Corollary}

\newcommand{\R}{\mathbb{R}}

\newcommand{\eps}{\varepsilon}
\renewcommand{\d}{\,\mathrm{d}} 

\newcommand{\psispace}{\mathcal{W}}
\newcommand{\rhogs}{{\rho_\mathrm{gs}}}
\newcommand{\rhogsprime}{{\rho'_\mathrm{gs}}}
\newcommand{\rhogspure}{{\rho_\mathrm{gs,pure}}}
\newcommand{\rhogsens}{{\rho_\mathrm{gs,ens}}}
\newcommand{\vHxc}{v_\mathrm{Hxc}}
\newcommand{\FHK}{F_\mathrm{HK1}}
\newcommand{\FHKpure}{F_\mathrm{HK1,pure}}
\newcommand{\FHKens}{F_\mathrm{HK1,ens}}
\newcommand{\FSD}{F^0_\mathrm{SD}}

\newcommand{\FLL}{F_\mathrm{LL}}
\newcommand{\FDM}{F_\mathrm{DM}}
\newcommand{\FCS}{F_\mathrm{CS}}
\newcommand{\FCSpure}{F_\mathrm{CS,pure}}
\newcommand{\FCSens}{F_\mathrm{CS,ens}}
\newcommand{\FHxc}{F_\mathrm{Hxc}}

\newcommand{\xx}{\mathbf{x}}
\newcommand{\rr}{\mathbf{r}}
\newcommand{\rrN}{\underline{\mathbf{r}}} 
\newcommand{\rrtwotoN}{\mathbf{r}_\perp}
\newcommand{\sigmaN}{\underline{\sigma}} 

\newcommand{\vv}{\mathbf{v}}

\newcommand{\vrep}{\ensuremath{v\text{-}\mathsf{rep}}}
\newcommand{\vreppure}{\ensuremath{\vrep_\mathrm{pure}}}
\newcommand{\vrepens}{\ensuremath{\vrep_\mathrm{ens}}}
\newcommand{\Nrep}{\ensuremath{N\text{-}\mathsf{rep}}}
\newcommand{\jpara}{\mathbf{j}^\mathrm{p}}

\newcommand{\loc}{\mathrm{loc}}

\DeclareMathOperator{\trace}{Tr}
\DeclareMathOperator{\conv}{conv}
\DeclareMathOperator{\dom}{dom}

\newcommand{\der}{\frac{\delta}{\delta \rho}}

\newcommand\changed[1] {#1}

\begin{document}

\author{Markus Penz}
\address{Basic Research Community for Physics, Innsbruck, Austria}

\author{Erik I. Tellgren}
\address{Hylleraas Centre for Quantum Molecular Sciences, Department of Chemistry, University of Oslo, Norway}

\author{Mih\'aly A. Csirik}
\address{Hylleraas Centre for Quantum Molecular Sciences, Department of Chemistry, University of Oslo, Norway}
\address{Department of Computer Science, Oslo Metropolitan University, Norway}

\author{Michael Ruggenthaler}
\address{Max Planck Institute for the Structure and Dynamics of Matter, Hamburg, Germany}

\author{Andre Laestadius}
\email{andre.laestadius@oslomet.no}
\address{Department of Computer Science, Oslo Metropolitan University, Norway}
\address{Hylleraas Centre for Quantum Molecular Sciences, Department of Chemistry, University of Oslo, Norway}

\title{The structure of the density-potential mapping\\Part I: Standard density-functional theory}

\begin{abstract}
  The Hohenberg--Kohn theorem of density-functional theory (DFT) is broadly considered the conceptual basis for a full characterization of an electronic system in its ground state by just the one-body particle density. Part I of this review aims at clarifying the status of the Hohenberg--Kohn theorem within DFT and Part II at different extensions of the theory that include magnetic fields. We collect evidence that the Hohenberg--Kohn theorem does not so much form the basis of DFT, but is rather the consequence of a more comprehensive mathematical framework. Such results are especially useful when it comes to the construction of generalized DFTs.
\end{abstract}

\maketitle
\tableofcontents

 
\section{Introduction} 

The theorem of \citeauthor{Hohenberg1964}~\cite{Hohenberg1964} \changed{(HK)} is usually presented as the theoretical justification of density-functional theory (DFT). It states that the one-body particle density uniquely (up to an additive constant) determines the scalar potential of a non-relativistic many-electron system in its ground state. 
The Mathematical analysis of ground-state DFT was pioneered by \citeauthor{Lieb1983}~\cite{Lieb1983}, using tools from convex analysis. In it, some important problems, especially in relation with differentiability of the involved functionals that map densities to energies, were left unanswered and remained as open questions.  
\citeauthor{Lammert2007}~\cite{Lammert2007} then demonstrated that the key functional of DFT is indeed non-differentiable, but it remained unclear to what \changed{extent} this threatens the foundations of DFT and its algorithmic realization, the Kohn--Sham scheme employed for practical calculations.
Regularization as a means to overcome non-differentiability has been applied to DFT \cite{Kvaal2014} (Section~\ref{sec:regDFT}) and its extension, current DFT (CDFT) \cite{KSpaper2018,MY-CDFTpaper2019}. 
The existence of functional derivatives through regularization also avoids the problem of $v$-representability that usually haunts DFT, i.e., that not every reasonable density is the solution to a certain potential (Section~\ref{sec:rep}).

A central result in this work is a very convenient and novel formulation of the \changed{HK} theorem that restructures it into two sub-theorems, HK1 and HK2 (Section~\ref{sec:HK}):
\begin{itemize}
\item \textbf{(HK1)} If two potentials share a common ground-state density then they also share a common ground-state wave function or density matrix.
\item \textbf{(HK2)} If two potentials share any common eigenstate and if that eigenstate is non-zero almost everywhere (a property that is guaranteed if the the unique-continuation property (UCP) holds; see Section~\ref{sec:UCP}) then they are equal up to a constant.
\end{itemize}
Combining HK1 and HK2, one obtains the classical \changed{HK} theorem and with it a well-defined density-potential mapping. The proof of HK1 will be shown to be immediate from just the formulation of ``ground-state energy''. Consequently, it is also easily attainable in an abstract or extended formulation of DFT (Section~\ref{sec:abstract-dens-pot}). The situation for HK2, on the other hand, is more complicated but, as will be demonstrated, it holds true with certain restrictions in the standard DFT setting. It is known \emph{not} to hold in paramagnetic CDFT~\cite{Capelle2002} and has, to the best of our knowledge, an unknown status in total CDFT.
In Part II of this review, we will exemplify how different DFTs follow this structure and, maybe more importantly, pinpoint why this route might fail.


After analyzing its basic structure, the status of the HK theorem within DFT is scrutinized. If only the ground state of a system is the matter of interest, a constrained-search approach seems to be sufficient for the formulation of DFT, and the usual type of constrained-search functional~\changed{\cite{percus1978,Levy79}} even \emph{implicitly includes} HK1 (Section~\ref{sec:functionals}). \changed{Besides being a mathematically more transparent formulation than the HK theorem, the constrained-search formalism is also a better starting point for deriving approximate density functionals. Nonetheless, the full HK theorem remains important for going beyond the bare minimum needed to set up a ground-state theory. For example, the HK theorem implies that the ground-state density determines not only the ground state but, by fixing the scalar potential, also all excited states. This becomes relevant when  thermostatistical properties are considered. Furthermore,} in order to be able to define the Kohn--Sham scheme (Section~\ref{sec:KS}), one actually demands \changed{more than} just the HK result, relying on differentiability of the energy functional that in turn would imply the \emph{whole} HK result (Section~\ref{sec:differentials}). Consequently, in a (Moreau--Yosida) regularized setting, the Kohn--Sham scheme can be rigorously formulated and even proven to converge in finitely many dimensions \cite{penz2019guaranteed,penz2020erratum,Kvaal2022-MY}, and HK becomes just a by-product.

Although we will do our best to orient the reader within the rich subject that is DFT, the scope of this review is limited. We will mainly focus on, in our opinion, matters closely related to the \changed{HK} mapping and properties of the exact functional(s). Many excellent reviews and textbooks are available on the subject \cite{vonBarth2004basic,burke2007abc,burke2012perspective,dreizler2012-book,eschrig2003-book,parr}. For the interested reader, we also point out the very recent article based on a round-table discussion~\cite{teale2022round-table}.

\section{Preliminaries}
\label{sec:prelim}

Density-functional theory is a\changed{n approach} to describe particles that obey the laws of quantum mechanics, but that avoids their full description by a wave function and instead switches to reduced quantities like the one-particle density. In its basic form discussed here, the focus is solely on the ground-state properties of the quantum system.
For the configuration space of a single particle we always choose $\R^3$ with the additional spin degree-of-freedom for spin-$\frac{1}{2}$ particles.
The Hamiltonian comprises three parts,
\begin{equation*}
H[v] = T + W + V[v],
\end{equation*}
relating to the kinetic energy, the Coulomb repulsion, and the external scalar potential, respectively.
The internal parts will be collected as $H_0 = T + W$. The kinetic-energy operator is $T = -\frac{1}{2}\sum_{i=1}^N \nabla_i^2$ in standard DFT, where atomic units are employed.
Notation-wise, we use small letters for one-body objects. The external potential contribution $V[v]$ is always defined from a one-body potential $v(\rr)$ and is of an additive form,
\begin{equation*}
V[v](\rrN) = \sum_{i=1}^N v(\rr_i),
\end{equation*}
where $\rrN=(\rr_1,\dots,\rr_N)$. For later reference we also define $\sigmaN = (\sigma_1,\dots,\sigma_N)$ for the spin degrees-of-freedom.
The full quantum-mechanical description of a system in its ground state is achieved by determining the eigenstate $\psi_0$ of $H[v]$ that has the correct symmetry and the lowest eigenvalue $E_0$ (ground-state energy),
\begin{equation}\label{eq:SE}
 H[v]\psi_0 = E_0 \psi_0.
\end{equation}
If such a lowest eigenstate is not unique, we speak of \emph{degeneracy}, a case that will often appear in the discussion below and that leads to several complicacies. Then a valid ground state can also be given as a statistical mixture of the pure ground states $\psi_k$ in the form of a density matrix $\Gamma = \sum_k \lambda_k \vert \psi_k\rangle \langle \psi_k \vert$ with $\lambda_k\in [0,1]$ and $\sum_k \lambda_k = 1$.
It is natural to require states of finite kinetic energy,
\begin{equation*}
   \langle \psi| T | \psi \rangle = \frac{1}{2}\sum_{i=1}^N \sum_{\sigmaN} \int_{\R^{3N}} |\nabla_i \psi|^2 \d\rrN < +\infty,
\end{equation*}
and we define the basic set for wave functions
\begin{equation*}
    \psispace = \{ \psi \mid \, \psi \; \text{anti-symmetric}, \,\langle \psi| T | \psi \rangle < +\infty  \}. 
\end{equation*}
In cases where density matrices $\Gamma$ are considered, we require $\psi_k\in\psispace$ for all their components.

The one-particle density of a given $\psi$ as the basic variable of standard DFT is
\begin{equation}\label{eq:def-rho}
\rho_\psi(\rr_1) = N \sum_{\sigmaN} \int_{\R^{3(N-1)}} |\psi|^2 \d \rrtwotoN,
\end{equation}
where we used the shorthand notation $\rrtwotoN=(\rr_2,\dots,\rr_N)$, and it is
$\rho_\Gamma(\rr) = \sum_k \lambda_k \rho_{\psi_k}(\rr)$ for a given mixed state $\Gamma$. Since $\Gamma$ already includes the squared wave function from Eq.~\eqref{eq:def-rho}, the mapping $\Gamma \mapsto \rho_\Gamma$ is linear. Note that whenever we talk about a ``density'', this will be assumed to be a map $\rho:\R^3 \to \R_{\geq 0}$ that is normalized to the particle number $N$, $\int\rho(\rr) \d\rr=N$, like it is automatically the case for $\rho_\psi$ and $\Gamma_\psi$ if $\psi,\Gamma$ are normalized to 1.

The density alone suffices to give an expression for the potential energy contribution. The resulting integral over the single-particle configuration space will be written like an inner product $\langle \cdot, \cdot \rangle$, to wit,
\begin{equation}\label{eq:v-rho-pairing}
\begin{aligned}
    \langle \psi| V[v] | \psi \rangle &= \sum_{i=1}^N \sum_{\sigmaN} \int_{\R^{3N}} v(\rr_i) |\psi|^2 \d \rrN \\
    &= N \sum_{\sigmaN} \int_{\R^3} v(\rr_1) \int_{\R^{3(N-1)}} |\psi|^2 \d \rrtwotoN \d \rr_1\\
    &= \int_{\R^3} v(\rr)\rho_\psi(\rr) \d \rr = \langle v,\rho_\psi \rangle.
\end{aligned}
\end{equation}
The notation $\langle v,\rho \rangle$ thus expresses a dual pairing between two $L^p$ spaces or a combination of such, one for densities and the other one for potentials. These density and potential spaces are the topic of the next section. Without going into technicalities, the space $L^p(\R^n)$, $1\leq p \leq \infty$, can be thought of as all functions $f(\rr)$ that have a finite $L^p$ norm
\begin{equation*}
    \|f\|_{L^p} = \left( \int_{\R^n} |f(\rr)|^p \d\rr \right)^{1/p} < \infty,
\end{equation*}
where in the case $p=\infty$ a supremum norm is employed instead.

\section{Representability of densities}
\label{sec:rep}

The notion of ``representability'' is ubiquitous and conceptually important in DFT. 
It generally refers to the situation that any density of a certain class comes from a well-defined construction.
Such a construction can simply be how a density is calculated from an $N$-particle wave function of finite kinetic energy following Eq.~\eqref{eq:def-rho} and we then call the density ``$N$-representable''. Or one demands that the density should be that of an actual ground-state solution of a Schr\"odinger equation with some given external potential $v$ and one calls it ``$v$-representable''. However, this definition of $v$-representability is a bit naive since the set of permitted potentials to choose from was not even specified \cite{Lammert2007}. One could argue that any potential that can be put into the Schr\"odinger equation should be considered, but then the dual pairing $\langle v,\rho\rangle $ appearing in Eq.~\eqref{eq:v-rho-pairing} between the spaces of densities and potentials might be ``lost'', which has consequences for the density functionals defined later in Section~\ref{sec:functionals}. So in order to talk about $v$-representability, we will first have to choose a basic density space that includes the $N$-representable densities.

The task of determining $N$-representable density classes was originally tackled by \citeauthor{gilbert1975}~\cite{gilbert1975} and \citeauthor{harriman1981}~\cite{harriman1981}. In the first work, differentiability of the density was required, whereas in the second work no further conditions on the density were assumed.
Here, we rely on the version by \citeauthor{Lieb1983}~\cite[Theorem~1.2]{Lieb1983} that gives the following class of $N$-representable densities,
\begin{equation*}
    \Nrep = \left\{ \rho \mid \rho(\rr)\geq 0, \smallint \rho \d \rr=N, \nabla \sqrt \rho \in L^2(\R^3) \right\}.
\end{equation*}
The benefit of the additional constraint $\nabla\sqrt \rho \in L^2$ is that one can always find a wave-function that not only gives the desired density but also has finite kinetic energy and is thus in $\psispace$ (and in addition is properly normalized). \citeauthor{Lieb1983}~\cite{Lieb1983} further showed that $\Nrep$ is convex and included in $X=L^1(\R^3)\cap L^3(\R^3)$. This space $X$ is the basic density space in terms of $L^p$ spaces, so by Eq.~\eqref{eq:v-rho-pairing} this automatically yields a corresponding potential space that is its dual, $X^*=L^{3/2}(\R^3)+L^\infty(\R^3)$. Any element $v\in X^*$ can thus be written as $v=v_1+v_2$ with $v_1\in L^{3/2}(\R^3)$ and $v_2\in L^\infty(\R^3)$. 
Potentials of Coulomb type, $v(\rr) = C r^{-1}$, $r =|\rr|$, are for example elements of this $X^*$ (by virtue of $\int_0^R |v(\rr)|^{3/2} r^2 \d r < \infty$ for any finite $R>0$ and $|v(\rr)|<\infty$ for $r>R$). 

The issue of ``$v$-representability'' is much more profound. To date there is no explicit description for the set of all $v$-representable densities $\vrep$. This issue is known as the ``$v$-representability problem''. We already noted that $\vrep$ should contain all densities that are a ground-state density for some potential $v\in X^*$. For a glimpse of what densities have to be included in this set we refer to the illustrative construction of \citeauthor{ENGLISCH1983}~\cite{ENGLISCH1983}. At this point one has to differentiate between several levels of $v$-representability. We defined \vrep{} as coming from a ground state of a Schr\"odinger equation with some given external potential $v$. Within DFT we usually consider two settings, the full system that contains a (Coulomb) interaction $W$ and the Kohn--Sham system that does not. So whenever we talk about $v$-representability, this can be amended by the attributes ``interacting'' or ``non-interacting'' and it is {\it not} obvious at this point if the two classes are equal, overlap, or are even disjoint. After all, the sets are not explicitly known. Within each class we also have the possibility of ground-state degeneracy. Then, instead of ground-state wave functions, the more general concept of density matrices 
comes into play. The resulting notions are then ``pure-state $v$-representability'' and ``ensemble $v$-representability''. In the second case such a density $\rho$ is then the convex combination of pure-state $v$-representable densities $\rho_k$ that come from the degenerate ground-states $\psi_k$ of $H[v]$, i.e., $\rho=\sum_k\lambda_k\rho_k$ ($\lambda_k \in [0,1], \sum_k\lambda_k=1$). In the first case only densities from pure states are allowed, but they might still individually come from a set of degenerate ground-state wave functions. It was demonstrated by \citeauthor{ENGLISCH1983}~\cite{ENGLISCH1983} by giving explicit examples that there are $N$-representable densities that are not ensemble $v$-representable (an obvious example is a density that vanishes on a set of positive measure, however, for more elaborate examples we refer to Section 3.2 in Ref.~\citenum{ENGLISCH1983}). \citeauthor{levy1982}~\cite{levy1982} and \citeauthor{Lieb1983}~\cite{Lieb1983} gave arguments that an ensemble $v$-representable density does not have to be pure-state $v$-representable. An explicit example for such a density $\rho \in \vrepens \setminus \vreppure$ was found within a finite-lattice system of cuboctahedral symmetry \cite{penz-DFT-graphs}. So we can symbolically note that
\begin{equation}\label{eq:vrep-Nrep-subset}
    \vreppure \subsetneqq \vrepens \subsetneqq \Nrep \subsetneqq X.
\end{equation}
In \citeauthor{Garrigue2021}~\cite{Garrigue2021} it was demonstrated that the set $\vreppure$ 
is path-connected. There are further topological relations between the sets appearing in Eq.~\eqref{eq:vrep-Nrep-subset} that are worth mentioning. Since every $\rho\in\vrepens$ is a convex combination $\rho=\sum_k\lambda_k\rho_k$ with $\rho_k\in\vreppure$, it holds
\begin{equation*}
    \vrepens \subseteqq \conv\vreppure \subsetneqq \Nrep \subsetneqq X,
\end{equation*}
where $\conv$ is the convex hull of a set. So while $\vreppure$ is definitely not convex because of the mentioned counterexamples, $\vrepens$ might still be (to our understanding this is not known). Lastly, $\Nrep$ is the closure of $\vrepens$ within $L^1\cap L^3$, which means that any $\rho\in\Nrep$ can be approximated arbitrarily well by densities in $\vrepens$ when distance is measured in the $L^1\cap L^3$-norm \cite[Theorem~3.14]{Lieb1983}. With the notion of the ``subdifferential'' from Section~\ref{sec:differentials}, this result can be established as a direct consequence of the Br{\o}ndsted--Rockafellar theorem \cite[Corollary~2.44]{Barbu-Precupanu}. Still, potentials that lead to densities that are arbitrarily close could be very far apart in the potential space $X^*$.
On the other hand, it has been suggested that $\vreppure$ is not dense in $\Nrep$ (see Conjecture~3.8 in Ref.~\citenum{garrigue2022building}).

\section{The Hohenberg--Kohn theorem}
\label{sec:HK}

The classical HK theorem \cite{Hohenberg1964} states the existence of a well-defined density-potential mapping for ground states.
For a given potential $v$,
\begin{equation}\label{eq:E-def}
\begin{aligned}
E[v] &= \inf\left\{ \langle \psi |H_0 + V[v]|\psi \rangle \mid \psi \in \psispace, \|\psi\|=1 \right\} \\
&= \inf\left\{ \langle \psi |H_0 | \psi \rangle + \langle v,\rho_\psi \rangle \mid \psi \in \psispace, \|\psi\|=1 \right\}
\end{aligned}
\end{equation}
is the \emph{ground-state energy} by the Rayleigh--Ritz variation principle. If a minimizer exists then $\psi$ and $\rho_\psi$ are the corresponding \emph{ground state} and \emph{ground-state density} that might not be unique in the case of degeneracy.
If a minimizer does not exist, there is still always a sequence $\psi_i$ in $\psispace$ with $\|\psi_i\|=1$ such that $\langle \psi_i |H_0 + V[v]|\psi_i \rangle$ converges to $E[v]$. 
In Eq.~\eqref{eq:E-def}, $v$ should be selected from a class that makes $E[v]$ bounded below. See \citeauthor{reed1975ii}, Section~X.2, for an extensive discussion on such potentials~\cite{reed1975ii}. A further demand on $v$ will later be that it guarantees a ground state that is non-zero (almost everywhere), a property needed in the proof of the second part of the HK theorem (HK2) below.

In Eq.~\eqref{eq:E-def} the problem of solving a partial-differential equation, the stationary Schr\"odinger equation \eqref{eq:SE}, has been transformed into a variational problem: finding a minimizer for Eq.~\eqref{eq:E-def}. The route backwards is also feasible and any such minimizer is also a distributional solution to the Schr\"odinger equation \cite[Theorem~11.8]{LiebLoss}. 

We will now demonstrate that simply by virtue of the structure of $E[v]$, where density and potential are combined in the term $\langle v,\rho \rangle$ that makes no explicit reference to the wave function while the remaining part $\langle \psi |H_0 |\psi \rangle$ (or $\trace(H_0\Gamma)$, if density matrices are used to describe the state) does not depend on $v$, we can already define a mapping from ground-state densities to ground-state wave functions or density matrices. This, then, is already half of a \changed{HK} theorem, that we will already give in a variant for ensemble $v$-representable densities.

\begin{theorem}[HK1]
Let $\Gamma_1$ be a ground state of $H[v_1]$ and $\Gamma_2$ a ground state of $H[v_2]$. If $\Gamma_1,\Gamma_2 \mapsto \rho$, i.e., if these states share the same density, then $\Gamma_1$ is also a ground state of $H[v_2]$ and $\Gamma_2$ is also a ground state $H[v_1]$.
\end{theorem}

\begin{proof}[Proof 1]
Since we assumed the existence of ground states $\Gamma_1,\Gamma_2$ for the potentials $v_1,v_2$, the infimum in Eq.~\eqref{eq:E-def}, when varied over density matrices, is actually a minimum.
Further, the potential-energy contribution $\langle v,\rho \rangle$ is fixed because $\rho$ is given and can be taken out of the minimum,
\begin{subequations}\label{eq:hk1-proof}
\begin{align}
    E[v_1] &= \min_{\Gamma'_1 \mapsto \rho} \trace(H_0 \Gamma'_1 )  + \langle v_1,\rho \rangle \nonumber\\
    & = \trace(H_0 \Gamma_1 )  + \langle v_1,\rho \rangle\\
    E[v_2] &= \min_{\Gamma'_2 \mapsto \rho} \trace(H_0 \Gamma'_2 )  + \langle v_2,\rho \rangle  \nonumber\\
    & = \trace(H_0 \Gamma_2 )  + \langle v_2,\rho \rangle
\end{align}
\end{subequations}
For completeness, we also give the same expression for a general $v$ in case the state is pure.
\begin{equation}\label{eq:hk1}
    E[v] = \min_{\psi \mapsto \rho} \langle \psi |H_0 | \psi \rangle  + \langle v,\rho \rangle \\
\end{equation}
Here, the notation ``$\Gamma \mapsto \rho$'' and ``$\psi \mapsto \rho$'' means variation over all states in $\psispace$ with density $\rho$. But the remaining minima in Eq.~\eqref{eq:hk1-proof} are then completely determined by the fixed ground-state density and we can always choose $\Gamma_1=\Gamma_2$ \changed{[primes removed]} as a valid ground state. Thus the density alone already defines the ground state, irrespective of the potential $v_1$ or $v_2$.
\end{proof}

As highlighted before, the above proof relies purely on the specific structure of the energy function $E[v]$ that allows the potential part to be taken as a separate, additive contribution that depends solely on the density. This idea is due to Paul E. Lammert (during discussion at the workshop ``Do Electron Current Densities Determine All There Is to Know?'' in Oslo, 2018). In contrast to this, the usual proofs of this part of the HK theorem additionally depend on the \emph{linear} structure of the density-potential pairing. Moreover, such proofs are almost always performed indirectly (\textit{reductio ad absurdum}), with a few notable exceptions  \cite{pino2007HK,Garrigue2018}. For completeness, we will give an additional, more traditional proof, yet one that is direct and does not work by raising a contradiction.

\begin{proof}[Proof 2]
By the variational principle, we have
\begin{equation*}
\begin{aligned}
    E[v_1] &= \trace( H[v_1] \Gamma_1) \leq \trace( H[v_1] \Gamma_2), \\
    E[v_2] &= \trace( H[v_2] \Gamma_2) \leq \trace( H[v_2] \Gamma_1).
\end{aligned}
\end{equation*}
Exploiting the shared density $\rho$, this may be written as
 \begin{equation*}
  \begin{split}
   E[v_1] & = \trace( H_0 \Gamma_1) + \langle v_1, \rho \rangle \\
   &\leq \trace( H_0 \Gamma_2) + \langle v_1 + v_2-v_2, \rho \rangle\\
   & = E[v_2] + \langle  v_1-v_2, \rho \rangle
   \end{split}
 \end{equation*}
and analogously as
 \begin{equation*}
   E[v_2] \leq E[v_1] + \langle  v_2-v_1, \rho \rangle.
 \end{equation*}
Combining the inequalities gives 
\[E[v_1] - E[v_2] = \langle  v_1-v_2, \rho \rangle
\]
and from 
\[
\trace( H[v_2] \Gamma_1) = \trace( H[v_1] \Gamma_1) - \langle  v_1-v_2, \rho \rangle
\]
that $\trace( H[v_2] \Gamma_1)=E[v_1]$. So $\Gamma_1$ is also a ground state of $H[v_2]$. Likewise, $\trace( H[v_1] \Gamma_2)=E[v_2]$, so $\Gamma_2$ is also a ground state of $H[v_1]$, as required.
\end{proof}

HK1 holds generally for mixed or pure ground states. The same proofs remain valid when the theorem is specialized to a statement about pure states $\Gamma_i = |\psi_i\rangle \langle \psi_i|$.
An immediate but maybe surprising consequence that is often referred to as the basis of DFT is that a ground-state density $\rhogs$ alone already determines a ground state. This result has been coined a \changed{\emph{weak HK-like result}} before~\cite{Tellgren2018} and it will be used to define the HK1 functionals on $\vreppure$ and $\vrepens$ in Eq.~\eqref{eq:F-HK-def} below. The ground state (\changed{associated} with $\rhogs$) is pure if $\rhogs \in \vreppure$ but has to be an ensemble if $\rhogs \in \vrepens \setminus \vreppure$.
Any state that is a minimizer in Eq.~\eqref{eq:hk1-proof} is really a ground state for \emph{all} potentials that share the same ground-state density. That all those potentials are in fact equal (up to a constant) is then the statement of HK2, the second part of the \changed{HK} theorem.
It will be formulated for eigenstates, in case of an ensemble we are free to just take any of its components.

\begin{theorem}[HK2]
If two potentials share any common eigenstate and if that eigenstate is non-zero almost everywhere, then the potentials are equal up to a constant.
\end{theorem}

\begin{proof}
If $v_1,v_2$ share a common eigenstate $\psi$ it holds
\begin{align*}
&(H_0 + V[v_1])\psi = E[v_1]\psi, \\
&(H_0 + V[v_2])\psi = E[v_2]\psi.
\end{align*}
Subtraction of the two equations and moving all potential parts that do not depend on $\rr_1$ to the right-hand side gives
\begin{equation}\label{eq:HK2-proof-step}
\begin{aligned}
(v_1(\rr_1)-v_2(\rr_1))\psi = \; & (E[v_1]-E[v_2])\psi \\ &- \sum_{i=2}^N (v_1(\rr_i)-v_2(\rr_i))\psi.
\end{aligned}
\end{equation}
Since we assumed $\psi$ non-zero almost everywhere, we can then divide by $\psi$ and get $v_1(\rr_1)-v_2(\rr_1) = \mathrm{constant}$ (almost everywhere) because the right-hand side does not depend on $\rr_1$.
\end{proof}

Since HK2 states that sharing \emph{any} common eigenstate for two potentials means that they are equal (up to a constant), this of course implies that the potentials share \emph{all} eigenstates because they yield exactly the same Hamiltonian (up to an additive constant that just shifts the spectrum).
The special requirement that the wave function is non-zero (almost everywhere) is guaranteed for a large class of potentials by the \emph{unique-continuation property (UCP) from sets of positive measure}. This property will be further discussed in Section~\ref{sec:UCP}. That zeroes (nodes) in the wave function \emph{are} still allowed on a set of measure zero is important here, since the fermionic many-particle wave functions will exhibit nodal surfaces when particle positions agree. Outside of the continuum setting, for example in finite-lattice systems, such a UCP is \emph{not} at hand and there are actual counterexamples to HK2, were two different potentials share a common eigenstate \cite{penz-DFT-graphs}.

The complete HK result is then obtained by combining the two theorems above. We will assume here that the potential is from the mentioned class that guarantees a non-zero ground state. We should remember that such or similar restrictions will always come into play if we want to show validity of a density-potential mapping in other settings. The statement will be formulated for densities in $\vrepens$, so it automatically holds for $\vreppure$ as well.

\begin{corollary}[HK]
If two potentials share a common ensemble $v$-representable ground-state density, then they are equal up to a constant.
\end{corollary}

\begin{proof}
By HK1 there is a density matrix $\Gamma = \sum_k \lambda_k \vert \psi_k \rangle \langle \psi_k \vert$ that is a ground state for both potentials. Since at least one $\lambda_k \neq 0$ and since the corresponding $\psi_k$ is a ground-state wave function for both Hamiltonians, the proof can be completed by HK2.
\end{proof}

This structuring into two separate theorems was already used in \citeauthor{kohn2004hohenberg}~\cite{kohn2004hohenberg}, just in the reverse order, for a brief argument about DFT with magnetization.
Historically, the HK theorem was first given only for the non-degenerate case and was only later extended to include degeneracy~\changed{\cite{ENGLISCH1983,Kohn1985}}. The proof presented here does not suffer from any limitation to non-degenerate ground states.

A final note is directed towards more general DFTs that will be briefly discussed in Section~\ref{sec:abstract-dens-pot} and especially in the forthcoming Part II of this review. For other types of potentials, like vector potentials, the statement in the HK theorem would not necessarily be that the potentials are equal ``up to a constant'', but for example ``up to a gauge transformation''. The set of gauge transformations that are possible without affecting the physical properties of the system then have to be specified within the respective theory.

\section{The unique-continuation property}
\label{sec:UCP}

In this section, we summarize some important results on the unique-continuation property (UCP) of solutions to the Schr\"odinger equation that is heavily used in the context of (mathematical formulation of) HK-type theorems. The current understanding is that the UCP cannot be avoided in a rigorous proof of a HK-type theorem.
The setting will be slightly more general than before and allow for dimensionality $d$ of the spatial part of the single-particle configuration space $\R^d$. The $N$-particle configuration space is then $\R^{n}$ with $n=dN$.

Roughly speaking, the desired UCP result states that under certain conditions on the potentials building up the operators $V$ and $W$ and if a solution $\psi$
to the (distributional) equation $H[v]\psi=0$ vanishes on a set of positive measure, then $\psi$ vanishes everywhere. That the right hand side is zero comes as no restriction here, since the energy $E$ can always be absorbed into the scalar potential $v$.
The usual literature on the UCP shows \emph{strong} UCP, which means that $\psi$ is assumed to \emph{vanish to infinite order} at a point $\rrN_0\in\R^n$ and then the statement follows.
A function $f(\rr)$ is said to vanish to infinite order at $\rrN_0\in\R^n$ if for all $k\ge 1$ there is a $c_k$ such that
\begin{equation*}
     \int_{|\rrN-\rrN_0|<\epsilon} |f(\rrN)|^2\,\mathrm{d}\rrN < c_k \epsilon^k
\end{equation*}
for every $0 <\epsilon <1$. 
Now a very convenient result by \citeauthor{Regbaoui}~\cite{Regbaoui} shows that the UCP on sets of positive measure actually follows from such a strong UCP if the potentials are in $L^{n/2}_\loc$. This work apparently built on \citeauthor{deFigueiredoGossez}~\cite{deFigueiredoGossez} that again rests on an early estimate for general Sobolev spaces by \citeauthor{Ladyzenskaya1968}~\cite[Lemma 3.4]{Ladyzenskaya1968}.
The result and its proof have been repeated in \citeauthor{lammert2018search}~\cite{lammert2018search}. For us that means that any strong UCP can also be used as a UCP on sets of positive measure which is the one needed for the proof of HK2. Yet the traditional strong UCP results, like most notably in \citeauthor{jerison-kenig1985}~\cite{jerison-kenig1985}, also give dimension-dependent constraints on the potentials like $L^{n/2}_\loc$, which approaches $L^\infty$ for growing particle number and is thus too restrictive for our use where singular potentials need to be considered. The saving idea recently came from \citeauthor{Garrigue2018}~\cite{Garrigue2018} and was also extended to more complex systems \cite{Garrigue2019,Garrigue2019b}: To take the special $N$-body structure of the potentials into account and thus avoid any dependence of the constraints on the particle number $N$.

\begin{theorem}[Garrigue's UCP]
Suppose that the potentials are in $L_{\loc}^p(\R^d)$ with $p>2$ for $d=3$ and $p=\max(2d/3,2)$ else. If a solution $\psi$ to the Schr\"odinger equation vanishes on a set of positive measure or if it vanishes to infinite order at any point, then $\psi=0$.
\end{theorem}

The most relevant case here is obviously $d=3$ which means that the potentials need to be in $L_{\loc}^p(\R^3)$ with $p>2$ but exactly $p=2$ is not enough yet. This clearly does not fit our potential space $X^* = L^{3/2}(\R^3) + L^\infty(\R^3)$, so while this UCP result is the best one available, it cannot be used for a HK2 theorem that covers the whole potential space of DFT in the formulation discussed here.
\citeauthor{Lieb1983}~\cite{Lieb1983} also remarked on the UCP in the context of the HK theorem, which ``is believed to hold'' for potentials in $X^*$, however in a weaker form that is not sufficient for the current purpose. So whenever we state that the HK holds in standard DFT, we actually mean under the given restrictions on the potentials.

\section{Hierarchy of density functionals}
\label{sec:functionals}

The first part of the \changed{HK} theorem, HK1, analogously holds in many different varieties of DFT (that will be explored in Part II), simply because its validity just depends on the form of the energy functional. 
HK1 then ensures that we can map from pure-state $v$-representable ground-state densities $\rhogspure$ to ground-state wave functions $\psi[\rhogspure]$ and from ensemble $v$-representable ground-state densities $\rhogsens$ to ground-state density matrices $\Gamma[\rhogsens]$. This makes it possible to define the HK1 functionals
\begin{subequations}\label{eq:F-HK-def}
\begin{alignat}{3}
    \label{eq:F-HKpure-def}
    &\FHKpure[\rhogs] = \langle \psi[\rhogs] |H_0 |\psi[\rhogs] \rangle &&\,\,\,\, \text{on}\; \vreppure  \\
    \label{eq:F-HKens-def}
    &\quad \text{and} \nonumber \\
    &\FHKens[\rhogs] = \trace (H_0 \Gamma[\rhogs]  ) &&\,\,\,\,\text{on}\; \vrepens
\end{alignat}
\end{subequations}
as the energy contribution only from the internal parts $H_0$ of the Hamiltonian. 
\changed{The universal nature of such functionals, being independent of any external $v$, justifies the usual attribution as \textit{universal} functionals.
It is then} possible to determine also the internal energy contributions for any state with density $\rhogs$ just from $\rhogs$.
To get the total ground-state energy \eqref{eq:E-def} with the help of the functional above, it is enough to vary over $v$-representable densities alone, instead of the much larger set of wave functions. We can write
\begin{equation}\label{eq:E-def-with-FHK}
\begin{aligned}
E[v] &= \inf\{ \langle \psi |H_0|\psi \rangle + \langle v,\rho_\psi \rangle \mid \psi \in \psispace, \|\psi\|=1 \} \\
&= \inf_{\rhogsprime}\{ \langle \psi[\rhogsprime],H_0\psi[\rhogsprime] \rangle + \langle v,\rhogsprime \rangle \} \\
&= \inf_\rhogsprime\{ \FHKpure[\rhogsprime] + \langle v,\rhogsprime \rangle \} \quad\text{on}\; X^*,
\end{aligned}
\end{equation}
or equivalently with $\FHKens$. 
We see already that there is a certain ambiguity in which density functional to use in the definition of $E[v]$. The other density functionals presented here will all have the property that they give the correct ground-state energy when applied in Eq.~\eqref{eq:E-def-with-FHK} which makes them all \emph{admissible} functionals \cite{kvaal-helgaker2015admissible}. Yet, they will differ with respect to their mathematical properties and we thus aim for the one with the best features.

The first problem here is that the densities to be considered in the variational problem are limited to those that are actual ground-state densities (\vrep), because else $\FHK[\rhogs]$ is left undefined, and we already learned in Section~\ref{sec:rep} that \vrep{} is not an explicitly characterized set.
Apart from that, HK1 just states the existence of a map $\rhogs \mapsto \psi$ or $\Gamma$ without giving any hints towards a constructive scheme. A first step to overcome these problems is to inspect Eq.~\eqref{eq:hk1}. 
This suggests the definition of another pair of density functionals that goes under the name of ``constrained search'',
\begin{subequations}\label{eq:F-CS-def}
\begin{alignat}{3}
    \label{eq:F-CSpure-def}
    &\FCSpure[\rho] = \inf_{\psi \mapsto \rho} \langle \psi| H_0 | \psi \rangle   &&\quad\text{on}\; \Nrep \; \text{and}\\
    \label{eq:F-CSens-def}
    &\FCSens[\rho] = \inf_{\Gamma \mapsto \rho} \trace (H_0\Gamma) &&\quad\text{on}\; \Nrep.
\end{alignat}
\end{subequations}
The domain is now the larger, convex, and explicitly defined $\Nrep$ in both cases. Note that the literature mostly denotes those functionals as $\FCSpure=\FLL$ (``Levy--Lieb'' \cite{Levy79,Lieb1983}) and $\FCSens=\FDM$ (from ``density matrix'' \cite{Lieb1983}). A recent, comprehensive study of these functionals can be found in \citeauthor{lewin2019universal}~\cite{lewin2019universal}.
Since the density is limited to the set $\Nrep$ that guarantees finite kinetic energy, the infima in Eq.~\eqref{eq:F-CS-def} are always attained, though not necessarily by a possible ground state (if $\rho$ is not $v$-representable), and can thus be replaced by minima in both cases~\cite[Theorem~3.3]{Lieb1983}. 
The convex combination of pure-state projections into density matrices translates to the functionals, so that $\FCSens$ is the convex envelope of $\FCSpure$ \cite[Proposition~18, the article treats DFT on a lattice but the statement and the proof of proposition remains exactly the same in the continuum case.]{penz-DFT-graphs}. This automatically ensures that $\FCSens$ is convex, a fact that can also be concluded from observing that $\Gamma \mapsto \rho$ linear \cite[Section~4.B]{Lieb1983}.

Since these density functionals appear in the optimization problem that determines the ground-state energy and density, like in Eq.~\eqref{eq:E-def-with-FHK}, convexity is of great importance because only for a convex functional can we be sure that identifying any \emph{local} minimum also means that a \emph{global} minimum has been found. So while we now know that $\FCSens$ is convex, the previous functional $\FHKpure$ does not even have a convex domain and therefore cannot be convex.
\citeauthor{levy1982}~\cite{levy1982} and \citeauthor{Lieb1983}~\cite{Lieb1983} also gave arguments for the non-convexity of $\FCSpure$. Since $\FCSens = \conv \FCSpure$, any density where $\FCSens[\rho] \neq \FCSpure[\rho]$ already shows non-convexity of $\FCSpure$. But this is equivalent to saying that $\rho$ is ensemble $v$-representable while it is \emph{not} pure-state $v$-representable, so $\rho\in\vrepens\setminus\vreppure$ \cite[Proposition~21]{penz-DFT-graphs}.

Note especially, that HK1 was necessary to define $\FHK$, but is not needed any more for the constrained-search functional $\FCS$. Being able to define a universal constrained-search functional, one that is independent of the potential like in Eq.~\eqref{eq:F-CS-def}, already fully facilitates the proof of HK1 and thus implies this result. A potential-independent constrained-search functional \emph{already implicitly includes HK1}. This implication was proven by \citeauthor{Levy79}~\cite{Levy79} along the lines of the usual HK proof and is mentioned in textbooks like \citeauthor{parr}~\cite[after their Eq.~(3.4.4)]{parr} and \citeauthor{tsuneda}~\cite[after Eq.~(4.5)]{tsuneda}. Speaking generally though, a constrained search is just as feasible if the constrained-search functional also depends on the external potential $v$ (although it would not be universal), so indeed this approach is more general than relying on HK1. Such a case turns up in CDFT when the current variable is the total current that itself depends on the vector potential (see Part II of this review for more on this).

By employing the constrained-search functional, the ground-state energy from Eq.~\eqref{eq:E-def} can now be rewritten again as
\begin{equation*}
\begin{aligned}
E[v] &= \inf\{ \langle \psi | H_0 |\psi \rangle + \langle v,\rho_\psi \rangle \mid \psi \in \psispace, \|\psi\|=1 \} \\
&= \inf\{ \FCSpure[\rho_\psi] + \langle v,\rho_\psi \rangle \mid \psi \in \psispace, \|\psi\|=1 \} \\
&= \inf_\rho \{ \FCSpure[\rho] + \langle v,\rho \rangle \} \quad\text{on}\; X^*,
\end{aligned}
\end{equation*}
or equivalently with $\FCSens$, where minimization is now performed over $\Nrep$.

When looking at non-interacting systems, the definitions of $\FHK$, Eq.~\eqref{eq:F-HK-def}, and $\FCS$, Eq.~\eqref{eq:F-CS-def}, involve only the kinetic-energy operator $T$ instead of $H_0$. We will then denote these functionals with a zero superscript, $\FHK^0, \FCS^0$, etc., that indicates that non-interacting systems are considered. A further functional then comes into play that is defined like $\FCSpure$, but where only Slater determinants are considered as wave functions. 
We define on $\Nrep$,
\begin{equation*}
    \FSD[\rho] = \inf_{\phi \mapsto \rho}\left\{ \langle \phi | T | \phi  \rangle \mid \text{$\phi$ is a Slater determinant} \right\}.
\end{equation*}
The usual name in the literature is $\FSD = T_S$.
This functional is of importance because it is the one used in Kohn--Sham theory which will be discussed in Section~\ref{sec:KS}. 
In their original article, \citeauthor{KS1965}~\cite{KS1965} implicitly set $\FSD=\FHKpure^0$ for all non-interacting pure-state $v$-representable densities, which has been noted to be wrong because of possible degeneracy \cite[Section~4.C]{Lieb1983}. On the other hand, for non-degenerate ground states $\phi$, which by necessity are always determinants in non-interacting systems, it holds that $\FSD[\rho_\phi]= \FCSpure^0[\rho_\phi]= \FHKpure^0[\rho_\phi]$, and else $\FSD\geq \FCSpure^0$. Nevertheless, for practical purposes, $\FSD$ usually takes up the role of the density functional when defining the energy functional in a non-interacting setting.

The transformation from any density functional $F_\bullet$ for an interacting system from above to the energy functional,
\begin{equation}\label{eq:E-def-from-any-F}
E[v] = \inf_\rho \{F_\bullet[\rho]+\langle v,\rho \rangle\} \quad\text{on}\; X^*,
\end{equation}
is called the \emph{convex conjugate} or Legendre--Fenchel transformation \cite[Section~2.1.4]{Barbu-Precupanu}.
There is also a way to reverse the transformation and we define
\begin{equation}\label{eq:F-def}
F[\rho] = \sup_v \{E[v]-\langle v,\rho \rangle\}\quad\text{on}\; X.
\end{equation}
This $F$ is the famous \emph{Lieb functional} \cite{Lieb1983}, yet another density functional, but this time the last one to be defined in standard DFT. It is the \emph{biconjugate} of any $F_\bullet$ considered before. Defined this way, both $E$ and $F$ are lower-semicontinuous and $E$ is concave while $F$ is convex and has the property $F\leq F_\bullet$ \cite[Proposition~2.19]{Barbu-Precupanu}. Actually, as a biconjugate, $F$ is the largest convex and lower semicontinuous functional that fulfills $F\leq F_\bullet$ which makes it the \emph{convex envelope} of $F_\bullet$. 
The domain is now the whole $X=L^1(\R^3)\cap L^3(\R^3)$, but automatically $F[\rho] = \infty$ for all densities that are not in $\Nrep$~\cite[Theorem~3.8]{Lieb1983}, while at the same time $F[\rho] < \infty$ if $\rho\in\Nrep$~\cite[Theorem~3.9 and the following Remark]{Lieb1983}. Let the \emph{effective domain} `$\dom$' of a convex functional be the elements from its domain where it is finite, then this means that $\dom F = \Nrep$.
Having reached $F$, it does not matter any more which (admissible) functional has been used in Eq.~\eqref{eq:E-def-from-any-F}, which means the convex envelopes of all the functionals above agree. Conversely, the Legendre--Fenchel transformation can also be utilized to go back from $F$ to $E$ \cite[Theorem~2.22]{Barbu-Precupanu},
\begin{equation}\label{eq:E-LF}
E[v] = \inf_\rho \{ F[\rho] + \langle v,\rho \rangle \}.
\end{equation}

We already noted that $F$ is convex and lower-semicontinuous, which are both important properties if we want to use the variational problem $E[v] = \inf_\rho \{ F[\rho] + \langle v,\rho \rangle \}$ to find a minimizing density. The same properties come into play when defining the minimizers by differentiation in Section~\ref{sec:differentials}.
From the definition of $F$ it follows directly that
\begin{equation*}
    E[v] \leq F[\rho] + \langle v,\rho \rangle,
\end{equation*}
a version of the Young inequality.
Equality in the above estimate holds if the density \emph{is} the ground-state density $\rhogs$ for the potential $v$,
\begin{equation*}
    E[v] = F[\rhogs] + \langle v,\rhogs \rangle.
\end{equation*}
For $\FCSpure$ the converse holds too: If $E[v] = \FCSpure[\rho] + \langle v,\rho \rangle$ then $\rho$ is a ground-state density within $\vreppure$ for the potential $v$ and further $\FHKpure[\rho] = \FHKens[\rho] = \FCSpure[\rho] = \FCSens[\rho]=F[\rho]$ \cite[Theorem~3.10]{Lieb1983}. But what about using the more general functional $F$ for the variational principle like in Eq.~\eqref{eq:E-LF}? Can we find a real ground state like this or will this variational principle yield additional artificial solutions because it is too general? Because it is the convex envelope of the other functionals, it cannot produce a functional value below the ground-state energy, but it could produce a minimizing density where there are no $v$-representable ground-state densities!
The problem is solved if we allow for ensembles of ground states: An ``amusing fact'' in \citeauthor{Lieb1983}~\cite[Eq.~(4.5)]{Lieb1983} gives $F=\FCSens$ on $\Nrep$, which effectively means $F=\FCSens$ since we can just set $\FCSens[\rho]=\infty$ outside of its domain $\Nrep$ to achieve equality globally on $X$. So any minimizer of $F+\langle v,\cdot \rangle$ is also one of $\FCSens+\langle v,\cdot \rangle$ and it is further the convex combination of ground states for the potential $v$. Consequently, when talking about ground states in the context of the functional $F$, we will always actually mean \emph{ensembles} of possibly degenerate ground states.

When comparing the functionals on $X$, we just set them to $\infty$ whenever we are outside their domains. The following hierarchy can be set up and is further laid out in Table~\ref{tab:DFT-functionals}.
\begin{equation*}
    F = \FCSens \leq \begin{array}{c} \FCSpure (\leq \FSD) \\ \FHKens \end{array} \leq \FHKpure.
\end{equation*}
Here, $\FSD$ appears in parentheses since it only comes into play in the non-interacting setting where we can perform the same type of transformations and have $F^0[\rho]$ and $E^0[v]$.

\begin{table}[h]
\begin{tabular}{lll|lll}
    \hline\hline
    $F_\bullet$ & \multicolumn{2}{l|}{convex} & domain & \multicolumn{2}{l}{convex} \\\hline
    $\FHKpure$ & no & & \vreppure & no \\
    $\FHKens$ & ? & & \vrepens & ? & \multirow{2}{*}{$\Downarrow$ cl} \\
    $\FCSpure$ & no & \multirow{2}{*}{$\Downarrow$ conv} & \Nrep & yes & \\
    $\FCSens$ & yes & & \Nrep & yes & \\
    $F$ & yes & & $L^1 \cap L^3$ & yes & \\
    \hline\hline
\end{tabular}
\caption{\label{tab:DFT-functionals}
The table shows the relations between the functionals discussed in Section~\ref{sec:functionals}. From $\FHKpure$ to $\FHKens$ the domain gets extended to $\vrepens$ while they agree on $\vreppure$. From $\FHKens$ to $\FCSpure$ the domain gets closed (cl) within $L^1\cap L^3$ and from $\FCSpure$ to $\FCSens$ the functional itself gets convexified (conv) while the domain remains the same. Finally, $F$ is just equal to $\FCSens$ on $\Nrep$.
}
\end{table}

\section{Density-potential mappings from differentials}
\label{sec:differentials}

In the previous section it was stated that in order to get the ground-state density of any system we have to find a solution to the variational problem
\begin{equation}\label{eq:E-F-equ2}
    E[v] = \inf_\rho \{ F[\rho] + \langle v,\rho \rangle \},
\end{equation}
now relying on the density functional $F$ from Eq.~\eqref{eq:F-def}.
To find the global minimum of a \emph{convex and lower-semicontinuous} functional we can perform differentiation, i.e., demand that the differential of $F[\rho] + \langle v,\rho \rangle$ with respect to $\rho$ must equal zero at the position of a ground-state density $\rhogs$.

\begin{figure}[ht]
\centering
\begin{tikzpicture}[scale=.65]
    \draw[->] (-4,-1) -- (4,-1)
        node[right] {$\rho$};
    \draw[domain=-4:0,smooth,variable=\x] plot ({\x},{(\x-2)*(\x-2)/8});
    \fill (0,.5) circle [radius=2pt];
    \fill (0,-1) circle [radius=2pt] node[below] {$\rho_0$};
    \draw[dotted,->] (0,.5) -- (0,4.5) node[above] {$\infty$};
    \draw[dashed] (-4,.5) -- (4,.5);
    \draw[dashed] (-4,2.5) -- (4,-1.5);
    \draw[dashed] (-4,-2.5) -- (4,3.5);
    \draw[dashed] (-1.412,-2.5) -- (1.88,4.5);
\end{tikzpicture}
\caption{Example of a convex and lower-semicontinuous function with a discontinuity at $\rho_0$ and some elements from the subdifferential displayed as linear continuous tangent functionals at $\rho_0$, represented by dashed lines.}
\label{fig:subdiff}
\end{figure}
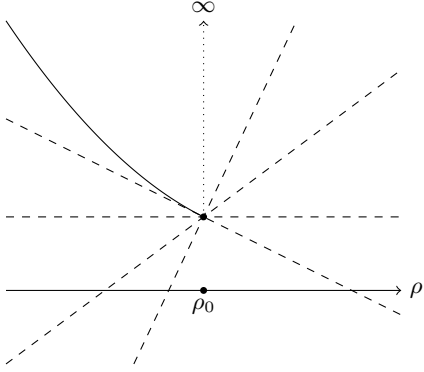

The suitable notion of differentiation here is the subdifferential $\underline\partial F$ that gives the \emph{set} of all linear continuous tangent functionals to a convex functional $F$ at a given density $\rho$,
\begin{equation*}
    \underline\partial F[\rho] = \{ v\in X^* \mid \forall \rho' \in X : F(\rho) \leq F(\rho')+ \langle v,\rho-\rho' \rangle \}.
\end{equation*}
It is always well-defined, since the set $\underline\partial F[\rho]$ can contain many elements, in case the functional $F$ has a kink (like the example shown in Fig.~\ref{fig:subdiff}), or can even be empty. Finally, if it contains exactly one element, we found a \emph{unique} potential yielding that ground-state density. In any case, the variational problem \eqref{eq:E-F-equ2} has a minimizer $\rhogs$ if and only if the following condition is fulfilled \cite[Proposition~2.33]{Barbu-Precupanu},
\begin{equation}\label{eq:v-subdiff}
    \left.\underline\partial(F[\rho] + \langle v,\rho \rangle)\right|_{\rho=\rhogs} \ni 0 \Longleftrightarrow \langle v,\cdot \rangle \in -\underline\partial F[\rhogs].
\end{equation}
In what follows, we identify $v$ with the functional $\langle v,\cdot \rangle$ whenever the context implies a functional on density space instead of a potential on configuration space, so Eq.~\eqref{eq:v-subdiff} can be written $v \in -\underline\partial F[\rhogs]$. The potential as the subdifferential of the density functional means that potentials $v$ are from the dual of the space of densities like already noted in Section~\ref{sec:rep}. This general principle is not always respected in more complex versions of DFT, as we will see in Section~\ref{sec:abstract-dens-pot} and discuss further in Part II of this review.

If the set $\underline\partial F[\rhogs]$ is non-empty then there is at least one potential $v\in X^*$ that yields the given ground-state density. The set of all densities where $\underline\partial F[\rhogs] \neq \emptyset$ is called the domain of the subdifferential, so it follows that $\dom \underline\partial F= \vrepens$. Note that by a theorem of convex analysis \cite[Corollary~2.44]{Barbu-Precupanu}, $\dom \underline\partial F$ is dense in $\dom F$, so $\vrepens$ is dense in $\Nrep$, a fact already expressed with $\Nrep$ being the closure of $\vrepens$ in Section~\ref{sec:rep}.

The meaning of a valid HK theorem for a class of densities is that they can all be mapped as ground-state densities back to a unique potential (modulo a constant) and consequently $-\underline\partial F[\rhogs] = \{v+c \mid c\in\R\}$. By eliminating the (physically unimportant) constant potentials from the potential space, the subdifferential of a $v$-representable density is precisely $-\underline\partial F[\rhogs] = \{v\}$ if the HK theorem holds. If, on the other hand, $F$ is assumed differentiable, then the directional derivative 
$-\der F[\rhogs] = v$ anyway always maps to a unique potential. One thus has a well-defined map from densities in $\vrep$ to the corresponding potentials, exactly the content of the HK theorem! But where did it enter? The HK theorem is here a \emph{consequence} from the assumption of differentiability of $F[\rho]$ at $v$-representable densities. The situation will be summarized diagrammatically in Section~\ref{sec:summary}.

Because any potential $v$ that we determine by Eq.~\eqref{eq:v-subdiff} will also be the maximizer in the conjugate variational problem
\begin{equation*}
    F[\rhogs] = \sup_v \{E[v]-\langle v,\rhogs \rangle\},
\end{equation*}
we can just as well say the same with the superdifferential of the concave functional $E$,
\begin{equation}\label{eq:E-superdiff}
    \left.\overline\partial(E[v'] - \langle v',\rhogs \rangle)\right|_{v'=v}=0 \Longleftrightarrow \rhogs \in \overline\partial E[v].
\end{equation}
The right hand side, $\rhogs \in \overline\partial E[v]$, means to find a density (or possibly many) that comes from a wave-function that minimizes the total energy including $v$. It is thus a conceptual shortcut to map from potentials to ground-state densities without any reference to an underlying wave function or Schr\"odinger equation.
The situation of a set $\overline\partial E[v]$ with more than one element is known from degeneracies of the Hamiltonian $H_0 + V[v]$, where different linearly independent ground states with eventually different densities all have the same eigenvalue.

We showed in this section the important role of the generalized concepts of sub/super\-differentials in the context of DFT, because indeed the functionals from Section~\ref{sec:functionals} can \emph{not} be assumed differentiable as \citeauthor{vanLeuuwen2003key}~\cite{vanLeuuwen2003key} has demonstrated for the $\FHK$ functionals and \citeauthor{Lammert2007}~\cite{Lammert2007} for $\FCS$. The reason for non-differentiability even of $\FCS$ is that at any $\rho$ the functional $F[\rho + \delta\rho]$ is infinite for various, arbitrarily small shifts $\delta\rho$ that lead out of $\Nrep$, even if the normalization of the density is kept constant. This happens by infinitely increasing the internal energy through tiny oscillations of the density. A possible way to prevent that is to limit the density space $X$ so that such shifts $\delta\rho$ are not possible any more and \citeauthor{Lammert2007}~\cite{Lammert2007} actually shows this for the Sobolev space $H^2(\R^3)$ when $\rho$ is also assumed $v$-representable. Another way is to establish a coarse-grained model for DFT in which $F$ really becomes differentiable and every density is ensemble $v$-representable \cite{lammert2010-coarse-grained}. 
In the following section, in accordance with the vast majority of the literature, we will assume functional differentiability of $F$ and consequently $v$-representability. This strong assumption can be justified \textit{a posteriori}, as discussed later in Section~\ref{sec:regDFT}, when a regularization procedure is applied.

\section{Linking to a reference system: the Kohn--Sham scheme}
\label{sec:KS}

In Section~\ref{sec:functionals} it was noted that a functional might be introduced for an interacting or a non-interacting system. This means the respective Hamiltonian has the internal part $T + \lambda W$ with $\lambda \in \{0,1\}$. We will now write $F^1$ and $F^0$ to differentiate clearly between those two situations.
We then introduce the difference functional $\FHxc = F^1-F^0$, which just corresponds to the internal-energy difference between the interacting and the non-interacting system and that will later be linked to the Hartree-exchange-correlation potential $\vHxc$. This potential effectively compensates for the Hartree-mean-field interaction as well as `exchange' and `correlation' effects. The idea behind introducing this auxiliary non-interacting system is that the energy difference between the (numerically tractable) non-interacting system and the (numerically unfeasible) interacting system is small and can be efficiently approximated.
Since the reference system is non-interacting, $\FSD$ can be employed for $F^0$ if degeneracy for the ground state does not have to be taken into account, like it was mentioned in Section~\ref{sec:functionals}, and this switchover is performed in most practical situations. Then the energy functional for the full system is
\begin{equation*}
\begin{aligned}
    E^1[v] &= \inf_\rho \{ F^1[\rho] + \langle v,\rho \rangle \} \\
    &= \inf_\rho \{ F^0[\rho] + \FHxc[\rho] + \langle v,\rho \rangle \} \\
    &= \inf_\phi \{ \langle \phi| T |\phi \rangle + \FHxc[\rho_\phi] + \langle v,\rho_\phi \rangle\}.
\end{aligned}
\end{equation*}
In the last step the variation is changed from $\Nrep$ densities to single Slater determinants $\phi$, the minimizer -- if it exists -- is then the Kohn--Sham Slater determinant. In order to link this to a partial differential equation for the orbitals $\varphi_i$ constituting $\phi$, the Kohn--Sham equation, variation of the energy expression above with respect to $\varphi_i$ is performed under the constraint that all the $\varphi_i$ stay normalized. This means $\rho_\phi(\rr) = \sum_{i=1}^N \sum_\sigma |\varphi_i(\rr\sigma)|^2$ always stays in $\Nrep$, but generally the issue of non-differentiability from Section~\ref{sec:differentials} remains. The resulting equation is a one-particle Schr\"odinger equation with effective potential $v_s$ and eigenstates $\varphi_i$,
\begin{equation}\label{eq:KS}
    \left( -\tfrac{1}{2}\nabla^2 + v_s(\rr)\right) \varphi_i(\rr\sigma) = \varepsilon_i \varphi_i(\rr\sigma).
\end{equation}
On the other hand this approach does not lead to the effective potential $v_s$ for the Kohn--Sham equation right away, but requires the additional, computationally challenging step of extracting the effective potential from the variation of $\FHxc$ with respect to the orbitals (OEP integral equation\cite{kummel2003optimized}).

To have a well defined $\FHxc[\rho] = F^1[\rho]-F^0[\rho]$, the $\rho$ must be both, interacting and non-interacting $v$-representable. Both systems then share the same ground-state density $\rho$ when the different external potentials
\begin{equation}\label{eq:v-vs}
    v \in -\underline\partial F^1[\rho] \quad \text{and} \quad v_s \in -\underline\partial F^0[\rho]
\end{equation}
are assigned to them.
That the density $\rho$ is simultaneously interacting \emph{and} non-interacting $v$-representable is tacitly assumed here, else one of the subdifferentials above is empty. This means that actually the $v$-representability problem from Section~\ref{sec:rep} shows up at this point. A purported solution \cite{gonis2014functionals,gonis2016reformulation,gonis2019interacting} rests on an ill-founded notion of differentiability where the functionals are extended to distributions, but with an incorrect application of the calculus of distributions (see, e.g., Eq.~(0.24) in \citeauthor{gonis2014functionals}~\cite{gonis2014functionals}).

The usual rationale of DFT is to assume that the potentials from Eq.~\eqref{eq:v-vs} exist and are unique (modulo a constant; after all the latter is the content of the HK theorem).
The difference $\vHxc = v_s-v$ is then known as the Hartree-exchange-correlation potential: what needs to be added to the fixed external potential $v$ in order to simulate all interactions in an non-interacting system. Note that such missing effects from interactions do not stem exclusively from the $W$-term in $F^1$, but also from the different kinetic energy contributions between the interacting and non-interacting system. Nevertheless, the usual understanding is that most of the kinetic energy contributions can already be captured by a non-interacting system (with an uncorrelated wave function) and that they thus practically cancel between $F^1$ and $F^0$ when we calculate
\begin{equation}\label{eq:vxc-subdiff-F}
    \vHxc[\rho] = v_s-v \in \underline\partial F^1[\rho] - \underline\partial F^0[\rho].
\end{equation}
At this point a problematic discrepancy is introduced, since the subdifferential is not linear and thus $\vHxc$ and $ \underline\partial\FHxc$ need not match.
If $\vHxc$ cannot be determined as $ \underline\partial\FHxc$ we are left with the necessity of individually solving the inverse problems $\rho \mapsto v$ and $\rho \mapsto v_s$ in Eq.~\eqref{eq:vxc-subdiff-F} for both systems, interacting and non-interacting. In practice this means one cannot benefit from finding good approximations to $\FHxc$ which are the most important elements of applied DFT.

A possible circumvention lies in a conceptual shift from describing a system in terms of energies to forces. The ground state is then characterized by a certain force-balance equation that can be equally found in non-equilibrium settings, just with an additional dynamical term \cite{tchenkoue2019force,stefanucci2013}. At a density that is simultaneously interacting and non-interacting $v$-representable and where the wave function has a sufficient regularity, the force-balance equation can be employed to derive $\vHxc$ as the solution of a Poisson equation instead of a functional derivative \cite{ruggenthaler2022force}. An alternative derivation for this was already given earlier using line integrals describing the work it takes to move an electron from infinity against the force field of the exchange-correlation hole charge \cite{harbola1991local,slamet1994force}.

Yet, we will proceed here for the sake of argument by \emph{assuming} differentiability for now. 
Since the functional derivative $\der$ is linear and it holds
\begin{equation}\label{eq:vxc-diff-F}
\begin{aligned}
    \vHxc[\rho] &= v_s-v = - \der F^0[\rho] +\der F^1[\rho]  \\
    &= \der (F^1[\rho] - F^0[\rho]) = \der\FHxc[\rho].
\end{aligned}
\end{equation}
Also, several important properties that the Hxc potential needs to have are automatically fulfilled when they are functional derivatives \cite{gaiduk2009potential_func_deriv}, which is especially relevant for functional approximations to $\vHxc$.

The \changed{Kohn--Sham} scheme is now introduced in order to find an unknown ground-state density $\rhogs$ of an interacting system by starting from an initial guess $\rho_0$ and by using $\vHxc$ (in practice a suitable approximation to it) as the connection between the interacting system and a non-interacting reference system.
To this end, rewrite Eq.~\eqref{eq:v-vs} with assumed differentiability as $\der F^1[\rhogs]+v=0$ and $\der F^0[\rhogs]+v_s=0$ and set the two equations equal,
\begin{equation*}
    \der F^1[\rhogs]+v = \der F^0[\rhogs]+v_s.
\end{equation*}
Now, apart from the fixed external potential $v$ of the interacting system, all variables in this equation still remain generally unknown: the effective potential of the non-interacting system $v_s$ and, especially, the density $\rho$ of both systems that we would like to determine. The trick lies in introducing sequences $\rho_i \to \rhogs$, $v_i \to v_s$ and define an update rule,
\begin{equation}\label{eq:vKSplus1}
    v_{i+1} = v + \der F^1[\rho_i] - \der F^0[\rho_i] = v + \vHxc[\rho_i].
\end{equation}
We see immediately that if $\rho_i$ has converged to the correct ground-state density $\rhogs$ of the interacting system, then $v + \der F^1[\rho_i]=0$ and the remaining equation tells us that indeed $v_{i+1}$ is the potential that gives the same density $\rhogs$ in the non-interacting system. The next step after Eq.~\eqref{eq:vKSplus1} in the Kohn--Sham iteration lies in determining the density $\rho_{i+1}$ that comes from $v_{i+1}$ in the non-interacting system (which is comparably easy achieved by solving the corresponding Kohn--Sham equation \eqref{eq:KS}) and then iterate. Convergence problems are a big issue within this iteration scheme and have been dealt with by either damping the iteration step from $\rho_i \to \rho_{i+1}$ to $\rho_i \to \rho_i + \mu ( \rho_{i+1} - \rho_i )$, $\mu \in (0,1)$, or mixing several of the previous steps $\{\rho_i\}$ into the result $\rho_{i+1}$ \cite{pulay1980convergence,cances2000damping,cances2021convergence}.  Guaranteed convergence has been studied and proven for the finite-lattice case \cite{penz2019guaranteed,penz2020erratum,Kvaal2022-MY} by combining an optimal damping step and a regularization technique \cite{Kvaal2014,Kvaal2022-MY}, the latter truly making $F$ differentiable and $E$ a strictly concave functional. This solves the problem of defining $\vHxc$ in Eq.~\eqref{eq:vxc-diff-F} and yields a curvature bound on $F$ that is needed for guaranteed convergence. The regularization method is briefly explained in Section~\ref{sec:regDFT} below. For the Kohn--Sham iteration in continuum DFT the convergence is still an open problem, a direct generalization of the finite-lattice case has been found to be insufficient \cite{penz2020convergence}.
\changed{In practical applications that suffer from convergence issues, imaginary-time propagation in time-dependent DFT has recently been found as a viable alternative to find a Kohn--Sham ground state~\cite{flamant2019imaginary}.}

\section{Density-potential mixing and regularized DFT}
\label{sec:regDFT}

The full HK theorem guarantees a unique inversion from densities to potentials, but the whole discussion, especially regarding the necessary conditions for showing HK2, probably already made us a little bit sceptical about its validity in different settings. We will thus introduce a method that always guarantees a bijective mapping, not between densities and potentials, but between \emph{quasidensities} (called \emph{pseudo-densities} in the original work on regularization \cite{Kvaal2014}) and potentials. The basic idea is simple: If for some reason we cannot guarantee a unique (injective) mapping from potentials to ground-state densities $v\mapsto\rho[v]$, meaning that different $v\neq v'$ map to the same $\rho[v] = \rho[v']$, then let us try it for $v \mapsto \rho_\eps[v] = \rho[v] - \eps v$, where at least in the previous example we would have $\rho_\eps[v] \neq \rho_\eps[v']$ for sure. One could argue that this could just as easily introduce new problems for injectivity, like having $v\neq v'$ such that $\rho_\eps[v] = \rho_\eps[v']$, but we will show in the following that this cannot be the case for the functionals considered here. Remember that the mapping $v \mapsto \rho[v]$ can be defined by the superdifferential of $E$, $\rho[v] = \overline\partial E[v]$, as explained in Eq.~\eqref{eq:E-superdiff}. So what is the corresponding functional $E_\eps$ such that $\rho_\eps[v] = \rho[v]-\eps v = \overline\partial E_\eps[v]$? The superdifferential retains the linear nature of a derivative if only concave functionals are added, so we can look for a convex functional $\phi$ such that $\overline\partial (-\phi)[v] = -\underline\partial \phi[v] = -v$. In a general space, such a question proves hard \cite{penz2020convergence}, but it is easy to see that in the usual space $L^2$ of square-integrable functions the norm square gives exactly what we need, $\phi[v] = \tfrac{1}{2}\|v\|^2 = \tfrac{1}{2}\langle v,v \rangle$. In any case, we have established $E_\eps = E - \eps \phi$ and $\rho_\eps[v] = \overline\partial E_\eps[v]$ with such a convex $\phi$. But in many cases, not only for the mentioned $L^2$ space, the functional $\phi$ is not only convex, but \emph{strictly convex}, meaning that any local minimizer is not only global but even unique. But this feature transfers to $E_\eps$ if $-\eps \phi$, as a \emph{strictly concave} functional, is added to $E$. Consequently, $E_\eps$ is also strictly concave and any maximizing potential in
\begin{equation}\label{eq:Feps-def}
    F_\eps[x] = \sup_v \{E_\eps[v]-\langle v,x \rangle\}
\end{equation}
is necessarily unique (not just up to a constant). This means we can always uniquely map $v \mapsto x = \rho[v]-\eps v$ and back. We wrote $x$ now to make clear that this is a quasidensity, a mixture between a density and its associated potential. As such it is neither necessarily normalized nor positive, just a general element of the density space, $x\in X$. 
By what we learned in Section~\ref{sec:differentials}, the \emph{quasidensity-potential mapping} can also be directly defined by $-\underline\partial F_\eps[x]= \{v\}$ for all $x$ without any  ``$v$-representability'' restriction for $x$. Consequently, the mapping is defined for all $x$ in the density space $X$ and thus bijective.

The whole maneuver of passing from $F$ to $F_\eps$ corresponds to a regularization strategy called Moreau--Yosida regularization \cite{Kvaal2014,Kvaal2022-MY} by which not only the concave $E$ transforms into a strictly concave $E_\eps$, but also the $F_\eps$ defined by Eq.~\eqref{eq:Feps-def} is finally differentiable if the spaces $X,X^*$ have some additional properties \cite[Theorem~9]{KSpaper2018}.
The only problem is that this requires the space $X$ to be reflexive, which it is not in our current formulation as introduced in Section~\ref{sec:rep}, since it includes the non-reflexive $L^1$ in its definition. So a different choice for the basic spaces, like $X=L^2$ on a bounded domain \cite{Kvaal2014} or $X=L^3$ as a larger alternative to our space\cite{KSpaper2018} has to be taken.

This section demonstrated how such a regularization that facilitates a unique (quasi)density-potential mapping can be used to fully circumvent any reference to the HK theorem. But to avoid confusion we will \emph{not} say that in a regularized setting the HK theorem ``holds'' even though a unique and well-defined (quasi)density-potential mapping exists.
\changed{It is interesting to note that the popular Zhao--Morrison--Parr method for density-potential inversion already implicitly employs Moreau--Yosida regularization and a limit procedure $\eps\to 0$~\cite{Penz2022ZMP}.}

\section{Abstract density-potential mapping}
\label{sec:abstract-dens-pot}

The presented form of HK1 allows for an abstraction and thereby for generalizations. Therein, the density is generalized to any system-inherent quantity that seems suitable to describe other system parameters that we are interested in. This could be the density together with the spin density, a current-quantity etc.
On the other side, we select a generalized form of the potential that enters the Hamiltonian and that is able to steer the ``density-quantity'' by coupling to it. Such a framework was developed in \citeauthor{KSpaper2018}~\cite{KSpaper2018}, building on Banach spaces and their duals for density and potential quantities. This enables us to employ the regularization technique from Section~\ref{sec:regDFT} to obtain a well-defined Kohn--Sham iteration scheme.

In order to be more concrete, let $\xx$ be the density quantity describing a state that will in general include many components, like different densities, currents etc., and $\vv$ the collection of external potentials acting on them. At this point we do not even assume that $\xx$ and $\vv$ have the same number or type of components like a dual structure between densities and potentials would impose. Instead of a linear pairing $\langle\vv,\xx\rangle$ for the coupling to the external potential we can introduce an arbitrary functional $f[\vv,\xx ]$.
Then the \emph{only} necessary condition left for an abstract HK1 is that the ground-state energy expression has the form
\begin{align}
    &\tilde F[\xx] = \inf_{\psi \mapsto \xx}\left\{ \langle \psi |H_0 | \psi \rangle \right\}, \nonumber\\
    &E[\vv] = \inf_\xx \{\tilde F[\xx] + f[ \vv,\xx ]\}.\label{eq:def-E-abstract}
\end{align}
Since $\tilde F[\xx]$ is independent of $\vv$, the critical argument in the first proof of HK1 still holds and thus two potentials that share a common $\xx$ in the ground state will also share a common ground-state wave function or density matrix. Consequently, HK1 is secured in any such formulation of DFT, while the situation for HK2 quite generally is more problematic. Even if the coupling between $\vv$ and $\xx$ that enters the energy functional in \changed{Eq.~\eqref{eq:def-E-abstract}} is linear like in $f[\vv,\xx] = \langle \vv, \xx \rangle$, the critical step \eqref{eq:HK2-proof-step} in the proof of HK2 will involve more degrees-of-freedom on the potential side and the argument may fail.

In the literature, the presented situation with linear coupling corresponds to what \citeauthor{Schoenhammer1995}~\cite{Schoenhammer1995} call $\{a\}$-functional theory.
Similarly, \citeauthor{higuchi2004arbitrary}~\cite{higuchi2004arbitrary,higuchi2004arbitrary2} allow for a more general choice of basic variables in DFT next to the usual density. 
\changed{\citet{xu2022extensibility} derived conditions that need to be fulfilled to also have a HK2 in such a general setting.}
\changed{One can then try and} extend DFT and the Kohn--Sham scheme systematically to predict further system parameters, if good approximative functionals can be found.

A first example would be the spin-resolved functional that has the usual one-particle density $\rho = \rho_{\uparrow} + \rho_{\downarrow}$ and the spin-density $\rho_{\uparrow} - \rho_{\downarrow}$ as basic variables, $\xx =  (\rho_{\uparrow} + \rho_{\downarrow},\rho_{\uparrow} - \rho_{\downarrow})$. An alternative possible choice would clearly be $\xx =  (\rho_{\uparrow},\rho_{\downarrow})$ \cite[Section~8.1]{parr}. The energy functional is $E[v] = \inf_\xx \{\tilde F[\xx] + \langle v,\rho \rangle\}$, with $v$ just the usual scalar potential that couples to the one-particle density $\rho$. The involved spaces for densities and potentials are not dual in this example, since they involve a different number of components. But by choosing an $\FHxc[\xx]$ that depends on the spin-resolved density, the Hxc-potential as its derivative (and with it the effective potential of the Kohn--Sham system) must be from the dual space of $\xx$ and thus include components that act on the different spin-components individually.

A second example is CDFT and its variants that will be thoroughly discussed in Part II of this review. The paramagnetic current density of a given state $\psi \in \psispace$ is defined as
\begin{equation*}
        \jpara_\psi(\rr_1) =N \sum_{\sigmaN} \int_{\mathbb R^{3(N-1)}} \mathrm{Im}\left\{ \psi^*\nabla_1\psi \right\} \d \rrtwotoN.
\end{equation*}
Then the amended density quantity is $\xx=(\rho,\jpara)$ which couples linearly to $\vv=(v + \frac{1}{2}|\mathbf{A}|^2, \mathbf{A})$ \cite{Vignale1987}.
%
Since by this the potential-energy contribution amounts exactly to the linear pairing $f[\vv,\xx] = \langle \vv,\xx \rangle$ that allows to define a potential-independent constrained-search functional, HK1 holds.

This means one can continue along the lines started in this work and try to generalize many concepts and results from above to such extended DFTs. This includes the definition of representable densities (Section~\ref{sec:rep}), different functionals (Section~\ref{sec:functionals}), functional differentiability (Section~\ref{sec:differentials}), setting up a Kohn--Sham scheme (Section~\ref{sec:KS}), as well as regularisation (Section~\ref{sec:regDFT}), since also there the existence of a full HK theorem was hardly ever assumed.

\begin{figure*}[ht!]
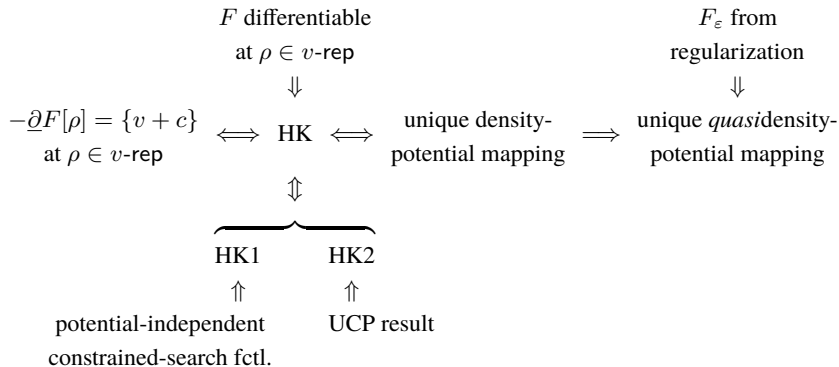

\begin{equation*}
\begin{array}{ccccccc}
    & \multicolumn{3}{c}{\text{$F$ differentiable}} &&&
        \text{$F_\eps$ from} \\
    & \multicolumn{3}{c}{\text{at $\rho \in \vrep$}} &&&
        \text{regularization} \\
    & & \Downarrow &&&& \Downarrow \\
    -\underline\partial F[\rho] = \{v+c\} &
        \multirow{2}{*}{$\Longleftrightarrow$} & 
        \multirow{2}{*}{HK} &
        \multirow{2}{*}{$\Longleftrightarrow$} &
        \text{unique density-} &
        \multirow{2}{*}{$\Longrightarrow$} &
        \text{unique \emph{quasi}density-} \\
    \text{at $\rho \in \vrep$} &&&&
        \text{potential mapping} &&
        \text{potential mapping} \\
    && \Updownarrow \\
    & \multicolumn{3}{c}{$\downbracefill$} \\
    & \text{HK1} && \text{HK2} \\
    & \Uparrow && \Uparrow \\
    \multicolumn{3}{c}{$\text{potential-independent}$} & \multicolumn{2}{l}{$\text{UCP result}$} \\
    \multicolumn{3}{c}{$\text{constrained-search fctl.}$}  
\end{array}
\end{equation*}
\caption{Logical implications between the different statements relating to a ``unique density-potential mapping'' and the \changed{HK} theorem in standard DFT.}
\label{fig:HK-structure}
\end{figure*}

\section{Summary}
\label{sec:summary}

We will give a brief summary of the structure of the density-potential mapping and its relation to the \changed{HK} theorem. Following the last section on abstract DFTs, at least the HK1 result does not only hold for standard DFT (that maps one-particle densities to scalar potentials), but it holds for all variants of DFTs that offer the required structures. This will be especially useful with foresight towards CDFT, the topic of the second part of this review.

In standard DFT, with a setting that yields the unique-continuation property that in turn prevents the ground-state density from being zero on a set of non-zero measure (Section~\ref{sec:UCP}) and due to the simple relation \eqref{eq:HK2-proof-step} in the proof of HK2, a full HK result can be established. In any higher DFT this proof strategy potentially fails. The status of HK1, on the other hand, is much less critical, since this result holds automatically whenever a \emph{potential-independent} (``universal'') constrained-search functional can be set up. But also in cases where the constrained-search functional depends on the external potential, a valid statement like in the HK theorem, that two potentials that share a common ground-state density are equal up to gauge changes, is still possible in general. The more general way how to think and talk about a HK result is by calling it a ``unique density-potential mapping'' and we explained how such a mapping can be established as the subdifferential of the density functional $F$ at $v$-representable densities. If the potentials in the resulting subdifferential are equal up to a gauge transformation, then this is just the HK result again. Assuming full differentiability of $F$ implies a one-element subdifferential, so there would not even be any room for gauge changes, and a unique density-potential mapping would be the result once more. This property of differentiability of the density functional $F$ is desirable also in the context of Kohn--Sham theory in order to be able to link the functional $F_\mathrm{Hxc}$ to the Hxc potential like in Eq.~\eqref{eq:vxc-diff-F}.

But since differentiability is \emph{not} a property of the usual DFTs, a regularization strategy was devised and briefly explained in Section~\ref{sec:regDFT}. This yields a unique \emph{quasi}density-potential mapping, where quasidensities are actually mixtures between ground-state densities and their potentials. The mixing parameter $\eps$ could be set to zero to retrieve the unregularized theory together with the problem of non-differentiability. The whole structure is laid out diagrammatically in Figure~\ref{fig:HK-structure}.

\section{Outlook}
\label{sec:outlook}

\changed{In this outlook, we first want to collect the problems that still remain open within the foundations of standard DFT and that will surely be the topic in upcoming works.
Considering Lieb's mathematical formulation of DFT, summarized above in Section~\ref{sec:rep} and \ref{sec:functionals}, there are two main issues.
Firstly, HK2 is guaranteed only for eigenstates that are non-zero almost everywhere, a property that is secured by the UCP explained in Section~\ref{sec:UCP}. But the potential space required for this does not cover all potentials from the Lieb setting and a sufficiently general UCP result is not available to date. Secondly, the issue of $v$-representability, explained in Section~\ref{sec:rep}, still remains open. While regularization as described in Section~\ref{sec:regDFT} formally allows us to circumvent this problem, it has not yet been put to practical use. Since the overlap between interacting and non-interacting $v$-representability is poorly understood, this has direct implications for Kohn--Sham theory. But even with $v$-representability assumed, convergence of the Kohn--Sham self-consistent field iterations in the standard setting is still an open problem. Both issues, availability of UCP and $v$-representability, relate to the function spaces for densities and potentials. Possibly, with a more refined choice of these spaces, full $v$-representability or even differentiability of $F$ might be achievable. However, it also cannot be ruled out that non-differentiability is fundamental to DFT.
}

\changed{
This non-differentiability of $F$, that has been repeatedly stressed in this work, implies that the exchange-correlation potential cannot be found as a functional derivative with respect to the density, as it is usually assumed in standard DFT.
Orbital-dependent functionals~\cite{KUMMEL_RMP80_3} can be formally viewed as relying on the HK1 map $\rho\mapsto\phi$ to obtain the Kohn--Sham wave function from a density. 
Non-differentiability of $F[\rho]$ might then be represented in the noninteracting wave function $\phi[\rho]$, which may benefit the functional approximations if they rely directly on the Kohn--Sham orbitals. 
%
The lack of differentiability also favors approaches based on forces instead of energies, as mentioned in Section~\ref{sec:KS}. However, practical functionals that are derived from this approach remain unexplored and there is still a dependence on $v$-representability.}

It is interesting to note which useful structures of DFT carry over to ``higher'' \changed{density-functional} theories, and in Part II we will discuss density-functional theory for systems involving magnetic fields. While one of its flavours, paramagnetic CDFT, already briefly discussed in Section~\ref{sec:abstract-dens-pot}, still allows for a constrained-search functional (HK1), the realization of a full density-potential mapping is highly problematic.
For this reason, in the classical formulation of paramagnetic CDFT \cite{vignale-rasolt-geldart1990} the HK2 result that different potentials lead to different ground states was just tacitly assumed with the words: ``Let $\psi$ and $\psi'$ be the two different ground states corresponding to the two sets of fields [$(v,\mathbf{A})$ and $(v',\mathbf{A}')$].'' Later, \citeauthor{Capelle2002}~\cite{Capelle2002} even found counterexamples to HK2 which shows that a density-potential mapping cannot be constructed in paramagnetic CDFT. But this clearly does not mean that in different versions of CDFT the density-potential mapping is impossible to achieve in general. A formulation utilizing the total current will be studied as well, but here the constrained-search functional would depend on $\mathbf{A}$ and thus HK1 is not available in the fashion as it was presented here. So while for paramagnetic CDFT the HK2 fails, for total (physical) CDFT already HK1 does not hold. Overall, the existence of a well-defined density-potential mapping in CDFT is still an open issue that will be considered in the second part of this review.

\section*{Acknowledgement}

EIT, MAC and AL thank the  Research Council of Norway (RCN) under CoE (Hylleraas Centre) Grant No. 262695, for AL and MAC also CCerror Grant No.~287906 and for EIT also ‘‘Magnetic Chemistry’’ Grant No.~287950, and MR acknowledges the Cluster of Excellence “CUI: Advanced Imaging of Matter” of the Deutsche Forschungsgemeinschaft (DFG), EXC 2056, project ID 390715994. AL and MAC was also supported by the ERC through StG REGAL under agreement No.~101041487.
The authors thank Centre for Advanced Studies (CAS) in Oslo, since this work includes insights gathered at the YoungCAS workshop ``Do Electron Current Densities Determine All There Is to Know?'', held July 9-13, 2018, in Oslo, Norway.

\section*{Bibliography}
%

\end{document}